\documentclass[11pt,a4paper]{article}
\usepackage{fullpage}
\usepackage[T1]{fontenc}
\usepackage[utf8]{inputenc}
\usepackage[colorlinks=true,allcolors=blue]{hyperref}
\usepackage{amsmath, amssymb, amsthm, amsfonts, graphicx,color,subcaption, enumerate, bm, cleveref, enumitem, array, tcolorbox, mathtools, booktabs}
\usepackage[table]{xcolor}
\usepackage{thm-restate}
\usepackage{berasans, beramono}
\usepackage[charter]{mathdesign}
\usepackage{tabularx}
\usepackage{multicol, multirow}
\usepackage{thmtools, thm-restate}
\usepackage{ragged2e}
\usepackage{lineno}
\usepackage[mathcal]{eucal}

\newtheorem{theorem}{Theorem}

\newtheorem{lemma}[theorem]{Lemma}
\newtheorem{corollary}[theorem]{Corollary}

\def\bdry{\partial\!}

\newcommand{\R}{\mathbb{R}}

\newcommand{\ecc}{\mathrm{ecc}}
\newcommand{\guess}{\ell}

\newcommand{\E}{\mathop{\mathbb{E}}}
\newcommand{\cA}{\mathcal{A}}
\newcommand{\cG}{\mathcal{G}}

\newcommand{\cM}{\mathcal{M}}
\newcommand{\cB}{\mathcal{B}}

\newcommand{\cS}{\mathcal{S}}

\newcommand{\cD}{\mathcal{D}}
\newcommand{\Rep}{\mbox{\sf Rep}}

\newcommand{\OO}{\widetilde{O}}

\newcommand{\ev}{\mathcal{E}}

\DeclareFontEncoding{LS1}{}{}
\DeclareFontSubstitution{LS1}{stix}{m}{n}
\DeclareSymbolFont{stixletters}{LS1}{stix}{m}{it}
\DeclareMathAccent{\cev}{\mathord}{stixletters}{"91}

\title{Real-weighted Diameter and Eccentricity of Minor-free and Bounded VC-dimension Graphs in Truly Subquadratic Time}
\author{Da Wei Zheng\thanks{This project has received funding from the Austrian Science Fund (FWF) grant  \href{https://www.doi.org/10.55776/I5982}{DOI 10.55776/I5982}. 
For open access purposes, the author has applied a CC BY public copyright license to any author-accepted manuscript version arising from this submission.}
}

\begin{document}
\maketitle
\begin{abstract}
We present the first truly subquadratic time algorithm to compute diameter and eccentricity in real-weighted directed graphs with constant distance VC-dimension and strongly sublinear-sized balanced separators. 
This runs in $O(n^{2-1/(2h-2)}\textrm{polylog}(n))$ time
for real-weighted $K_h$-minor-free digraphs.

Prior to this work, truly subquadratic time computation of diameter was only known for real-weighted planar graphs, while extensions to broader classes like minor-free graphs were restricted to unweighted settings.
In particular, existing algorithms that use VC-dimension [Ducoffe, Habib, Viennot; SICOMP 2022][Le, Wulff-Nilsen; SODA 2024][Chan, Chang, Gao, Le, Kisfaludi-Bak, Zheng; FOCS 2025] work with small integer weights, but do not naturally generalize to real weights.
We overcome this barrier by introducing a randomized search-to-decision reduction, demonstrating that VC-dimension is a sufficiently powerful tool in the real-weighted regime.
\end{abstract}

\section{Introduction}

\renewcommand{\arraystretch}{1.2}
\newcommand{\cc}{yellow!20}
\begin{table}[t]
\centering
\begin{tabular}{l @{\hspace{5em}} c c c l}
\toprule
\textbf{Graph class} & \textbf{Runtime} & \textbf{Real-weighted?} & \textbf{} \\
\midrule
\multirow{2}{*}{Planar}
  & $\OO(n^{11/6})$ & Y &  \cite{Cabello19} \\
  & $\OO(n^{5/3})$  & Y &  \cite{GawrychowskiKMS21} \\
\midrule
\multirow{6}{*}{$K_h$-minor-free}
  & $\OO(n^{2-1/2^{O(h)}})$        & N &  \cite{DucoffeHV22} \\
  & $\OO(n^{2-1/O(h^2)})$           & N &   \cite{LeW24} \\
  & $\OO(n^{2-1/(3h-2)})$           & N & \cite{KarczmarzZ25} \\
  & $\OO(n^{2-1/(2h-2)})$           & N &  \cite{ChanCGKLZ25} \\
  & \cellcolor{\cc} $\OO(n^{2-1/(2h-2)})$           & \cellcolor{\cc}\textbf{Y} & \cellcolor{\cc} \Cref{cor:minorfree} \\
\midrule
\multirow{3}{*}{
\shortstack{$O(n^\beta)$-separator  \\
g.d. VC-dim $d$}}
  & $\OO(n^{2-(1-\beta)/2^{O(d)}})$        & N &  \cite{DucoffeHV22} \\
  & $\OO(n^{2-1/(2d)})$ & N & \cite{ChanCGKLZ25}\\
  & \cellcolor{\cc} $\OO(n^{2-(1-\beta)/ d})$        & \cellcolor{\cc}\textbf{Y} & \cellcolor{\cc} \Cref{thm:diam} \\
\midrule
g.d. VC-dim $d$
& $\OO(mn^{1-1/(2d)})$ & N & \cite{ChanCGKLZ25}\\
\bottomrule
\end{tabular}
\caption{Summary of runtime guarantees for diameter and eccentricity algorithms in directed graphs with $n$ vertices and $m$ edges.
Generalized distance VC-dimension is abbreviated as ``g.d. VC-dim''. All bounds are for directed graphs. Note that all graphs besides g.d. VC-dim d graphs are sparse and have $m=O(n)$.}
\label{tab:graph-algorithms}
\end{table}
For a sparse graph $G$ with $n$ vertices and $O(n)$ non-negative real-weighted edges, it is possible to compute the diameter of $G$ in $\OO(n^2)$ time\footnote{Throughout this paper, we will use $\OO(\cdot)$ to suppress polylogarithmic factors of $n$.} by running Dijkstra's algorithm from every vertex of the graph.
However, even in unweighted undirected graphs, computing the exact diameter in truly subquadratic time is impossible assuming the Strong Exponential Time Hypothesis (SETH) \cite{RodittyW13}.
This negative result has motivated the study of diameter in structured graph classes where truly subquadratic-time diameter is achievable, such as
bounded-treewidth graphs~\cite{CabelloK09,BringmannHM20}, planar graphs~\cite{Cabello19,GawrychowskiKMS21}, bounded genus graphs~\cite{KlukPPS25}, $K_h$-minor-free graphs~\cite{DucoffeHV22,LeW24,KarczmarzZ25}, and geometric intersection graphs~\cite{DurajKP24, ChanCGKLZ25}.

Two major paradigms have emerged for breaking this quadratic barrier for structured graph classes. The first is abstract Voronoi diagrams for planar graphs that were pioneered by Cabello~\cite{Cabello19} and led to the first truly subquadratic-time diameter algorithm in planar graphs that ran in $\OO(n^{11/6})$ time. The running time was later improved to $\OO(n^{5/3})$ by Gawrychowski, Kaplan, Mozes, Sharir, and Weimann~\cite{GawrychowskiKMS21}.
While abstract Voronoi diagrams easily extend to real-weighted and directed settings, the approach strongly relies on planarity and is likely only generalizable to bounded-genus graphs.%

The second paradigm revolves around \emph{VC-dimension}.
Ducoffe, Habib, and Viennot~\cite{DucoffeHV22} demonstrated that it is possible to get a truly subquadratic-time algorithm in undirected unweighted $K_h$-minor-free graph with \emph{distance VC-dimension}, the VC-dimension of \emph{neighborhood balls} centered at every vertex of the graph.
Le and Wulff-Nilsen~\cite{LeW24} generalized a different bounded VC-dimension set system --\emph{distance profiles} -- studied by Li and Parter~\cite{LiP19} in planar graphs to obtain results for directed $K_h$-minor-free graphs. 
Bounds on both distance VC-dimension and distance profiles were unified 
by Karczmarz and Zheng~\cite{KarczmarzZ25} with \emph{generalized distance VC-dimension}\footnote{Originally \cite{KarczmarzZ25} called this the \emph{multiball} set system. The name \emph{general distance VC-dimension} was given in \cite{ChanCGKLZ25}.}, that paved the way for improvements in more graph classes~\cite{ChanCGKLZ25}.

A striking discrepancy remains: Techniques using abstract Voronoi diagrams work in directed and real-weighted planar graphs~\cite{Cabello19,GawrychowskiKMS21,CharalampopoulosGLMPWW23}, while techniques using VC-dimension results do not work well in real-weighted graphs.
While the VC-dimension bounds directly generalized to real-weighted settings, and \cite{LeW24} extended the VC-dimension results to directed graphs, the existing algorithmic machinery did not seem to work well with real weights.
The two limited results that did exist for weighted $K_h$-minor-free graphs either required integer weights with logarithmic dependence on the maximum weight~\cite{DucoffeHV22}, or achieved truly subquadratic space for a distance oracle at the cost of superquadratic construction time~\cite{KarczmarzZ25}.
This motivates the central question of this paper:
\textit{Is VC-dimension useful for truly subquadratic-time computation in real-weighted graphs?}


\subsection{Our Contribution}

We answer our motivating question in the affirmative by giving the first algorithm for real-weighted diameter computation in $K_h$-minor-free digraphs. 
Our result generalizes to any monotone graph family with strongly sublinear-sized balanced separators, a family equivalent to graphs of \emph{polynomial expansion}~\cite{DvorakN16}.
We show that if the generalized distance VC-dimension is bounded, we can compute the vertex eccentricities (and thus diameter) of the graph in truly subquadratic time.
See \Cref{tab:graph-algorithms} for a comparison of our results to previous algorithms.

\begin{restatable}{theorem}{thmdiam}\label{thm:diam}
Let $\cG_\beta$ be a monotone graph class where every $n$-vertex graph in $\cG_\beta$ has $O(n^\beta)$-sized balanced separators.
If the generalized distance VC-dimension of an $n$-vertex real-weighted digraph $G \in \cG_\beta$ is $d$, there is a randomized Las Vegas algorithm that computes all vertex eccentricities in $\OO(n^{2-(1-\beta)/d})$ time.
\end{restatable}
$K_h$-minor-free graphs are a special subclass of graphs with polynomial expansion as they are a monotone graph class with $\beta=1/2$ (i.e. they have $O(\sqrt{n})$ sized separators). They also have generalized distance VC-dimension at most $h-1$~\cite{KarczmarzZ25}, so the corollary follows.
\begin{corollary}\label{cor:minorfree}
There is a Las Vegas randomized algorithm that can compute the diameter and all vertex eccentricities of a real-weighted $n$-vertex $K_h$-minor-free digraph $G$ in $\OO(n^{2-1/(2h-2)})$ time.
\end{corollary}

The previous algorithm of \cite{DucoffeHV22} computed all vertex eccentricities with a decision procedure: given a guess $\ell_u$ on the eccentricity of every vertex $u$ in the graph $G$, outputs whether $\ell_u$ is an over- or under-estimate of the true eccentricity of $u$.
Using this procedure, it is easy enough to binary search for the eccentricities of every vertex in graphs where edges have small integer weight, or to obtain $(1+\varepsilon)$-approximations.
Unfortunately, this binary search fails to yield exact results for real-weighted graphs since there can be exponentially many distances.

To bypass this limitation, we augment the decision procedure to output \textit{\textbf{random witnesses}}. To be precise, if a guess $\ell_u$ for the eccentricity of $u$ is an underestimate, the algorithm should additionally output a random \emph{witness} $v$ whose distance from $u$ is larger than $\ell_u$. 
By transforming the simple decision oracle to a random witness producing oracle, we reduce the search for the real-valued eccentricity to a small number of augmented oracle calls. See \Cref{sec:randwit} for the full details.
We present the entire algorithm in \Cref{sec:eccalg} by extending and significantly simplifying the decision algorithm of \cite{DucoffeHV22}.

\section{Preliminaries}
Throughout this paper, we will consider a directed graph $G = (V,E)$ with $n=|V|$ vertices and positive real weights on all edges.
For two vertices $u,v\in V$, we denote the length of the shortest directed path from $u$ to $v$ in $G$ as $d_G(u\to v)$.
For convenience, we will assume throughout this paper that there is a unique shortest path between any two vertices $u, v\in V$. 
Note that for any graph $G$ that does not have unique shortest paths, we may apply (tiny) perturbed weights to every edge in $G$ to create a new graph where shortest paths are unique and shortest paths in this new graph are also shortest paths in $G$. This can be done in many ways, e.g., with the Isolation Lemma~\cite{ValiantV86}.
For any $k\in \R$, define the \emph{$k$-neighborhood ball} (or $k$-neighborhood for short) of a vertex $u\in V$ to be the set of vertices that are reachable from $u$ and at distance at most $k$ from $u$, denoted by $N^k_G[u] = \{v\in V\mid d_G(u\rightarrow v) \le k\}$. 

The \emph{eccentricity} of a vertex $u\in G$ is the maximum distance from $u$ to any other vertex, and is denoted by $\ecc(u) = \max_v d_G(u\to v)$. The \emph{diameter} of the graph $G$ is the maximum eccentricity of any vertex, i.e. $\max_v \ecc(v)$.
For simplicity, all graphs throughout this paper will be strongly connected unless otherwise specified.
If the graph $G$ is not strongly connected, then there may not be a directed path from $u$ to $v$ and the distance $d_G(u\to v)$ is not well-defined. In this case, it is possible to define the eccentricity of a vertex $u$ to be the distance to the farthest \emph{reachable} vertex. 

\section{Random Witnesses for Eccentricity}
\label{sec:randwit}

In this paper, we will focus on the problem of computing eccentricity of every vertex, as diameter can be easily computed by taking the maximum eccentricity among all vertices. We begin by considering an algorithm to compute the eccentricity of a single vertex.

\paragraph*{Single Vertex Eccentricity}
Fix one $u\in V$, and suppose we have a guess for its eccentricity $\ell_u$.
Suppose we had an algorithm $\cA$ for deciding if the eccentricity of $u$ was at most $\guess_u$. 
Suppose our guess was an underestimate, i.e. that $\guess_u < \ecc(u)$. 
It is reasonable to expect the algorithm $\cA$ to give a witness $w\in V$ that proves that $\guess_u$ is smaller than the eccentricity of $u$, meaning that $\guess_u < d_G(u\to w)$. 
We can use $d_G(u\to w)$ as a new guess $\ell'_u = d_G(u\to w)$ for the eccentricity, and run the algorithm $\cA$ again.

We use the following notation to denote the set of vertices at distance strictly greater than $k$ away from vertex $u$.
\[N^{> k}_G[u] = \{v\in V\mid d(u\to v) > k\}\]
$N^{> \ell_u}_G[u]$ is exactly the set of witnesses that $\cA$ could have chosen, so $w\in N^{> \ell_u}_G[u]$. When we run the algorithm $\cA$ on our new guess $\ell'_u$, we would get a new witness from $N^{>\ell'_u}_G[u]$. 
Observe that $N^{>\ell'_u}_G[u] \subset N^{>\ell_u}_G[u]$ (the inclusion is strict). 
It is possible that $|N^{>\ell'_u}_G[u]| = |N^{\ell_u}_G[u]|-1$, so very little progress was made.
This could occur if we were unlucky or the algorithm $\cA$ was adversarial.
On the other hand, if the algorithm $\cA$ gave us a random witness, we would expect that
\[\E\left[|N^{>\ell'_u}_G[u]|\right]=\frac{1}{2}|N^{>\ell_u}_G[u]|.\]
It is easy to see that by repeatedly finding a random witness, we would expect to find the vertex of maximum distance from $u$ in $O(\log |N^{>\ell_u}_G[u]|) = O(\log n)$ many rounds. It is folklore that this bound additionally holds with high probability\footnote{By high probability, we mean with probability at least $1-1/n^{c}$ for some fixed constant $c$.}. 
For completeness, we present the following proof based on backwards analysis~\cite{HarPeled2018Backward}.
\begin{lemma}\label{lem:minchange}
Start with the set $S_0 = \{1,\dots, n\}$. 
Consider the following process that at time $t\ge 1$, samples an element $x_t$ uniformly randomly from $S_{t-1}$ and computes: 
\[ S_t:= \{y\in S_{t-1} \mid y > x_t\} \quad \text{(Only keep elements strictly larger than $x_t$) }\]
This process terminates with the empty set by time $t=c_0\log n$ with probability $1-1/n^{\Omega(c_0^2)}$.
\end{lemma}

\label{ap:minchange}
\begin{proof}[Proof of \Cref{lem:minchange}]
We will couple this process with a random permutation.
Consider drawing elements of $S_0$ without replacement to get a permutation $\pi = (\pi_1, \dots, \pi_n)$.
Call an element of $\pi_i$ a \emph{record} if $\pi_i = \max\{\pi_1, \pi_2, \dots, \pi_{i-1}\}$.
The process described in the lemma samples precisely the record elements of $\pi$. One can view this as a rejection sampling implementation of the process where elements of the permutation that are smaller than the current record are skipped.

Define $\ev_i$ as the event $\pi_i$ is the maximum of $\{\pi_1,\dots, \pi_i\}$, and the variable $X_i=1$ if $\ev_i$ happens and $X_i=0$ otherwise. Observe that $\E[X_i] = 1/i$, and if we let $X = \sum_{i=1}^n X_i$ it follows that $\E[X] = \sum_{i=1}^n \E[X_i] = \sum_{i=1}^n1/i = \Theta(\log n)$. 
We claim that the random variables $\{X_i\}_{i=1}^n$ are in fact independent.
Assuming this claim and using standard Chernoff bounds \cite[Equation 1.10.5]{Doerr20}, the lemma follows:
\[P[t \ge c_0\log n] = P[X \ge c_0\log n] \le \exp\left(-\Omega(c_0^2\log n)\right) = 1/n^{\Omega(c_0^2)}. \]

We now prove the claim that $\{X_i\}_{i=1}^n$ are independent random variables, so the events $\{\ev_i\}_{i=1}^n$ are independent. Consider any collection of these events $\left\{\ev_{i_1}, \ev_{i_2}, \dots, \ev_{i_k}\right\}$ with indices $i_1 < i_2 <\cdots < i_k$. 
View the permutation $\{\pi_1, \dots, \pi_n\}$ as generated in reverse order (i.e. first $\pi_n$ was sampled from $S$, then $\pi_{n-1}$ is sampled from the remaining $n-1$ elements, and so on). The final event $\ev_{i_k}$ happens if and only if it is chosen as the maximum among $i_k$ elements and is independent of all events $\ev_j$ with $j< i_k$. Thus,
\[ P[\ev_{i_k}\mid \ev_{i_1}\cap\ev_{i_2}\cap\cdots\cap\ev_{i_{k-1}}] = P[\ev_{i_k}]. \]
It follows from induction that:
\begin{align*}P[\ev_{i_k}\cap\ev_{i_1}\cap\ev_{i_2}\cap\cdots\cap\ev_{i_{k-1}}] &= P[\ev_{i_k}\mid \ev_{i_1}\cap\ev_{i_2}\cap\cdots\cap\ev_{i_{k-1}}]\cdot P[\ev_{i_1}\cap\ev_{i_2}\cap\cdots\cap\ev_{i_{k-1}}] \\
&= P[\ev_{i_k}]\cdot P[\ev_{i_1}\cap\ev_{i_2}\cap\cdots\cap\ev_{i_{k-1}}] 
= \prod_{j=1}^kP[\ev_{i_j}] 
\end{align*}
so we conclude that $\{\ev_i\}_{i=1}^n$ are independent. 
\end{proof}

\paragraph*{Multiple Vertex Eccentricity}
Suppose we had a guess $\ell_u$ on the eccentricity of every vertex $u\in V$ that is an underestimate (so $\ell_u<\ecc(u)$), as well as an algorithm $\cA$ for deciding the guess of every $\ell_u$ simultaneously and also returns a random witness for every $\ell_u$. By a simple union bound with \Cref{lem:minchange}, we will find the eccentricity of \emph{every} vertex in at most $O(\log n)$ rounds with high probability.
\begin{lemma}\label{lem:minchange-multi}
If $n$ many processes (not necessarily independent) described in \Cref{lem:minchange} are all run simultaneously, the slowest one will finish in at most $O(\log n)$ many rounds with high probability.
\end{lemma}
In the next section, we show an algorithm that simultaneously returns a random witness for every vertex that runs in truly subquadratic time provided the graph has strongly sublinear-sized balanced separators and bounded generalized distance VC-dimension (for definitions see \Cref{ssec:alg-prelim}).

\section{Eccentricity Algorithm}
\label{sec:eccalg}
\subsection{Definitions}
\label{ssec:alg-prelim}
We begin with definitions of graph separators, VC-dimension, and low average crossing spanning paths that are important tools.
\paragraph*{Graph classes, balanced separators, and $r$-divisions.}
A graph class is \emph{monotone} if it is closed under taking subgraphs. For a graph $G$ with $n$ vertices, a \emph{balanced separator} is a set of vertices $S$ such that every connected component of $G \setminus S$ has size at most $2n/3$.

Let $\cG_\beta$ be a monotone graph class such that any $n$-vertex graph $G\in \cG_\beta$ has $O(n^\beta)$-sized balanced separators. 
Note that all such graphs $G\in \cG_\beta$ are sparse~\cite[Lemma 12]{DvorakN16}, i.~e. if $G$ has $n$ vertices, $G$ has $O(n)$ edges.
In $n$ vertex $K_h$-minor-free graphs, recent work of \cite{BonnetKLLM25} show that balanced $O(\sqrt{n})$-sized balanced separators can be computed in linear time.
In general graphs, recent work of \cite{KolmogorovSJ26} show that balanced separators of nearly-optimal size can be efficiently computed by a black-box reduction to maximum flow, which can be solved in $m^{1+o(1)}$ time in $m$-edge graphs \cite{ChenKLPGS25}. 
\begin{lemma}[Balanced separators~\cite{KolmogorovSJ26}] \label{lem:balanced-cuts}
For any graph $G\in \cG_\beta$ with $n$ vertices, a balanced separator of size $\OO(n^\beta)$ can be constructed in $n^{1+o(1)}$ time.
\end{lemma}

A \emph{piece} $P\subseteq V$ of a graph is a subset of vertices such that the induced graph $G[P]$ is connected.
An \emph{$r$-division} is a set of pieces $P_1, P_2, \dots, P_k$ each with at most $r$ vertices such that $V = \bigcup_{i=1}^k P_i$, and for every two pieces $P_i$ and $P_j$, the vertices $P_i \setminus P_j$ and $P_j \setminus P_i$ have no edges between them (they may share some vertices).
For a piece $P$, the vertices with edges to $G\setminus P$ we refer to as \emph{boundary vertices}, while all other vertices in $P$ we refer to as \emph{interior vertices}. Let $\bdry P$ denote the set of boundary vertices of a piece $P$. 
The \emph{total boundary} of the $r$-division is the quantity $\sum_{i=1}^k |\bdry P_i|$. 
The idea of $r$-divisions was originally introduced by Frederickson~\cite{Frederickson87} by using sublinear-sized balanced separators for planar graphs%
\footnote{The original algorithm of Frederickson used weighted $0$--$1$ separators to reduce the size of the boundary for each piece. We skip this step and only guarantee that the total boundary is small.}.
 The algorithm is now standard: recursively find balanced separators until each component has size at most $r$. The algorithm is efficient as \Cref{lem:balanced-cuts} runs in near-linear time as all graphs in $\cG_\beta$ are sparse.
\begin{lemma} \label{lem:r-division}
For any graph $G\in \cG_\beta$ with $n$ vertices, an $r$-division with total boundary $\OO(n/r^{1-\beta})$ can be constructed in $n^{1+o(1)}$ time.
\end{lemma}

\paragraph*{VC-dimension} A set system $(X, \cS)$ is a ground set $X$ together with a collection $\cS \subseteq 2^{X}$ of subsets of the ground set $X$.
For a subset $Y\subseteq X$, we denote by $\cS_{| Y}$ to be the set of $\cS$ restricted to $Y$, that is 
$\cS_{|Y} = \{Y\cap S\mid S\in \cS\}$.
The set $Y$ is \emph{shattered} by $\cS$ if $|\cS_{|Y}| = 2^{|Y|}$.
The \emph{shatter function} of a set, denoted by $\pi_{(X,\cS)}(h)$, is the largest number value of $|\cS_{\mid Y}|$ among all size-$h$ subsets $Y\subseteq X$.
The \emph{shatter dimension} of the set system $(X, \cS)$ is the smallest value of $d$ such that $\pi_{(X,\cS)}(h) = O(h^d)$ for all $h$.
The \emph{VC-dimension} of the set system $(X, \cS)$ is the size of the largest subset $Y\subseteq X$ that can be shattered by $\cS$.
The \emph{dual set system} is the set system $(\cS, X^*)$ where the ground set is indexed by $\cS$, and each element $x\in X$ induces a set $x^*\in X^*$ consisting of the ground set elements $x^* = \{S\in \cS \mid  S\ni x\}$.
The \emph{dual VC-dimension} is the VC-dimension of the dual set system, and the \emph{dual shatter dimension} is the shatter dimension of the dual set system.
We list two known facts about VC-dimension.
\begin{lemma}\label{lem:vcdim-facts}
For a set system $(X, \cS)$ with VC-dimension $d$:
    \begin{enumerate}
        \item[(1)] For any $Y\subseteq X$, the VC-dimension of $(Y, \cS_{\mid Y})$ is at most $d$.
        \item[(2)] For any $\cS'\subseteq \cS$, the VC-dimension of $(X, \cS')$ is at most $d$.
        \item[(3)] The dual VC-dimension is at most $2^{d+1}$.
        \item[(4)] (Sauer-Shelah Lemma) The shatter dimension of $(X, \cS)$ is at most $d$.
    \end{enumerate}
\end{lemma}

\paragraph*{Generalized distance VC-dimension}
Recall that for a vertex $u$ in a graph $G$, the notation $N^k_G[u]$ denotes the set of vertices within a distance $k\in \R$ from $u$.
We denote the set of all neighborhood balls as $\cB_G = \{N^k_G[v]\mid v\in V, k\in \R\}$. 
The \emph{distance VC-dimension} of a (weighted) digraph $G=(V,E)$ is the VC-dimension of the set system $(V, \cB_G)$. \emph{Generalized $k$-neighborhood balls} $M^k_G[u] = \{(v, r)\in V\times \R \mid d(u\to v) \le r + k\}$ were studied by~\cite{KarczmarzZ25} and considered in~\cite{ChanCGKLZ25}. One way to think about $M^k_G[u]$ is that it contains $(v, r)$ for all values of $r$ where $r \ge d(u\to v)-k$. The set of all generalized neighborhood balls in $G$ is denoted by $\cM_G = \bigcup_{k,v} M_G^k[v]$.
This allows us to define the \emph{generalized distance VC-dimension} of a weighted digraph $G = (V, E)$ as the VC-dimension of the set system $(V\times \R, \cM_G)$. 

One benefit of this definition is that the generalized distance VC-dimension bound is an upper bound on the dual VC-dimension of $(V, \cB_G)$, which we call the \emph{dual distance VC-dimension}. This is used by~\cite{ChanCGKLZ25}, although its proof is not explicitly stated. For completeness we prove it here.
\begin{corollary}\label{cor:dualvc}
Let $G$ be a weighted digraph, and let $\cev{G}$ denote the graph obtained from $G$ by reversing the direction of every edge. 
If $\cev{G}$ has generalized distance VC-dimension $d$, then the dual distance VC-dimension of $G$ is at most $d$.
\end{corollary}
\begin{proof}
The dual distance VC-dimension is the VC-dimension of the set system $(\cB_G, V^*)$.
Consider any vertex $u\in V$ and the set 
\[{M}_{\cev{G}}^0[u]=\{(v,r)\in V\times \R\mid d_{\cev{G}}(u\rightarrow  v)\le r\} = \{(v,r)\in V\times \R\mid d_{{G}}(v\to u)\le r\}.\] 
Let $\cM^0 = \{{M}_{\cev{G}}^0[u]\mid u\in V\}$.
Observe that $(v, r) \in M_{\cev{G}}^0[u]$ if and only if $u \in N_G^r[v]$. As each element of $\cB_G$ is an $r$-neighborhood ball $N^{r}[v]$, we can index this ball by $(v, r)$. This gives a one-to-one correspondence between $(\cB_G, V^*)$ and $(V\times \R, \cM^0)$.
Observe that the set system $(V\times \R, \cM^0)$ has VC-dimension at most $d$ by \Cref{lem:vcdim-facts}(2) as $\cM^0 \subset \cM_G$ and by assumption $\cev{G}$ has bounded distance VC-dimension $d$. 
\end{proof}
Generalized distance VC-dimension is natural to consider as both $K_h$-minor-free (di)graphs~\cite{KarczmarzZ25} and geometric intersection graphs~\cite{Chang0024} have bounded generalized distance VC-dimension. 
\begin{theorem}[\cite{KarczmarzZ25}]\label{thm:minorvc}
$K_h$-minor-free graphs have generalized distance VC-dimension at most $h-1$.
\end{theorem}
\begin{theorem}[\cite{Chang0024}]\label{thm:diskvc}
Disk intersection graphs have generalized distance VC-dimension at most $4$.
\end{theorem}

\paragraph*{Low average crossing spanning paths} 
Consider a set system $(X,\cS)$ with $|X|=n$ and $|\cS| = m$, and an ordering $\lambda$ of the elements of $X$. For an $S\in \cS$, define \emph{intervals} to be the maximal contiguous subsequences of $S$ in $\lambda$. We denote the set of these intervals (whose union is $S$) as $\Rep_\lambda(S)$ and call this the \emph{$\lambda$-interval representation of $S$}. 
The ordering is said to \emph{cross} (or \emph{stab}) $S$ if $\lambda_i \in S$ but $\lambda_{i+1} \not \in S$, or $\lambda_i \not\in S$ but $\lambda_{i+1} \in S$ (this ordering is sometimes called a \emph{stabbing path}). Note that the number of crossings is proportional to $|\Rep_\lambda(S)|$. A \emph{low average crossing spanning path} is an ordering $\lambda$ where each set $S\in \cS$ is on \emph{average} crossed a small number of times. This is weaker than the original notion of low crossing spanning paths introduced by Chazelle and Welzl~\cite{ChazelleW89,Welzl92} that guarantees \emph{every} set is crossed a sublinear number of times. 
It was proven in \cite{DurajKP24} and \cite{ChanCGKLZ25} that these low average crossing spanning paths can be efficiently constructed with high probability.
\begin{lemma}[Lemma 2.10 from \cite{ChanCGKLZ25}] \label{lem:lowcrossing}
Let $(X,\cS)$ be a set system with $|X| = n$ and $|\cS| = m$. 
If $(X, \cS)$ has dual shatter dimension at most
d, there is a low average crossing spanning path $\lambda$ with 
\[ \sum_{S\in \cS} 
|\Rep_\lambda(S)| = \OO(mn^{1-1/d}).\]
Furthermore, if it takes $T_0$ time to enumerate all elements of any set $S\in \cS$, then $\lambda$ can be constructed in $\OO(T_0n^{1/d})$ time with high probability.
\end{lemma}

\subsection{Eccentricities Algorithm}
Here we describe the algorithm for computing the eccentricity of every vertex of a strongly connected $n$-vertex graph $G=(V,E)$ with $G \in \cG_\beta$ where $\cG_\beta$ is a monotone graph class with $O(n^\beta)$-sized balanced separators and $G$ has generalized distance VC-dimension $d$.

We first give a high level overview of the algorithm. The algorithm computes an $r$-division (for a value of $r$ we will specify later) and tries to compute the eccentricity of a vertex $u$ in a piece $P_i$. The eccentricity will either be achieved strictly within the piece $P_i$, or it will be achieved by some path through a boundary vertex of $P_i$. The first case is easily handled since $P_i$ has at most $r$ vertices and we can afford to compute distances within $G[P_i]$. The second case can be handled as we can afford to compute shortest paths from all boundary vertices. This leads to the following equality:
    \begin{equation} \label{eq:nbhd}
    N_G^{\ell_u}[u] = N^{\ell_u}_{G[P_i]}[u] \cup \bigcup_{v\in \bdry P_i} N_G^{\ell_u-d_G(u\to v)}[v].
    \end{equation}
Note that if $c$ is negative, then $N_G^c[v]$ is defined to be the empty set.
This union can be constructed efficiently with a low average crossing spanning path $\lambda$, and there are at most $r$ additional elements of $N^{\ell_u}_{G[P_i]}[u]$ that can be manually added.
We can use a single (global) low average crossing spanning path as we construct all these neighborhood balls.

\paragraph*{Algorithm for all vertex eccentricities}        
Here we present the algorithm for computing all vertex eccentricities. We remark that steps 5(a)-5(c) closely resemble the algorithm of \cite{DucoffeHV22}, and 5(d) is the main contribution of this paper.
\begin{enumerate}
    \item Compute an $r$-division resulting in pieces $P_1, P_2, \dots, P_k$ with total boundary size $\OO(n/r^{1-\beta})$.
    \item Let $V_{in}\subset V$ denote the set of internal vertices of the $r$-division. For every piece $P_i$, compute the shortest path between every internal vertex $u\in V_{in}$ and all other vertices in $P_i$ where distances are computed within the graph $G[P_i]$.
    \item For every vertex $v$ that is a boundary vertex of any piece, compute the distance from $v$ to every vertex of $V$ in $G$.
    \item  Initialize a vector of guesses to eccentricity $(\ell_u)_{u\in V_{in}}$ where $\ell_u=0$ for all $u\in V_{in}$. 
    \item Repeat the following procedure for all vertices simultaneously (that either finds a random witness disproving that the eccentricity is $\ell_u$, or concludes that $\ecc(u) = \ell_u$).
    \begin{enumerate}
        \item 
        Consider the neighborhood balls centered at boundary vertices $v\in \bdry P_i$ with distance $\ell_u-d(u\to v)$ for every internal vertex $u\in V_{in} \cap P_i$. Define the following (multiset) of balls:
        \[
        \cB^{out}_i = \left\{N_G^{\ell_u-d(u\to v)}[v] \biggm\vert u\in V_{in}\cap P_i, v\in \bdry P_i\right\} 
        \]
        Consider the multiset $\cD = \bigcup_{i=1}^k \cB^{out}_i$.
        Compute a low average crossing spanning path $\lambda$ for the set system $(V, \cD)$.
        \item For each piece $P_i$, each boundary vertex $v\in \bdry P_i$, and each neighborhood ball $N^t_G[v]\in \cB_i^{out}$,
        compute the $\lambda$-interval representation $\Rep_\lambda(N^t_G[v])$ in increasing order of radius $t$.
        \item For each piece $P_i$, and each internal vertex $u\in V_{in}\cap P_i$, use \Cref{eq:nbhd} to compute the $\lambda$-interval representation $\Rep_\lambda (N_G^{\ell_u}[u])$.
        \item For each internal vertex $u\in V_{in}$, if $\Rep_\lambda (N_G^{\ell_u}[u])$ is an interval consisting of all vertices, we conclude that the eccentricity is $\ell_u$. Otherwise, we can sample a random witness $w\in V\setminus N_G^{\ell_u}[u]$ by using the interval representation $\Rep_\lambda (N_G^{\ell_u}[u])$, and update $\ell_u$ to be the value
        \begin{equation}d_G(u\to w)= \min\left\{d_{G[P_i]}(u\to w),\min_{v\in \bdry P_i} d_{G[P_i]}(u\to v) + d_G(v\to w)\right\}\label{eq:dist}\end{equation}
        where $d_{G[P_i]}(u\to w) = \infty$ if $w$ is not a vertex of $P_i$.
    \end{enumerate}
\end{enumerate}

\paragraph*{Correctness}
We can easily read off the eccentricity of the boundary vertices since we computed all distances from boundary vertices in step 3. For a non-boundary vertex $u\in V_{in}\cap P_i$, the initial guess of $\ell_u = 0$ is an obvious lower bound for eccentricity and $\ell_u$ stays a lower bound for eccentricity as it changes in step 5d. Furthermore, it can only increase.
The correctness of \Cref{eq:nbhd} means that accurate neighborhood balls are computed at the end of step 5c. The argument that paths must lie entirely within some $P_i$ or must intersect a boundary vertex also shows correctness of \Cref{eq:dist} for efficient computation of distances to random witnesses.

\paragraph*{Runtime}
Step 1 takes $n^{1+o(1)}$ time by \Cref{lem:r-division}. 
Step 2 takes $\OO(r)$ time using Dijkstra's algorithm for each internal vertex $u\in P_i$ since the induced graph $G[P_i]$ has $r$ vertices and $O(r)$ edges, for a total of $\OO(nr)$ time for all internal vertices.
Step 3 takes $\OO(n)$ time for performing Dijkstra's algorithm per boundary vertex for a total of $\OO(n)\cdot \OO(n/r^{1-\beta}) = \OO(n^2/r^{1-\beta})$ time.
Step 5 is repeated at most $O(\log n)$ times with high probability by \Cref{lem:minchange-multi}.
Step 5a takes $\OO(n^{1+1/d})$ time using \Cref{lem:lowcrossing} since we can enumerate any of the desired neighborhood balls in $\OO(n)$ time with Dijkstra's algorithm.
Note that this running time is independent of the number of sets $m = |\cD|$ that we had in our set system.
\[m = |\cD| = \sum_{i=1}^k|\cB_i^{out}| = \sum_{i=1}^k|P_i|\cdot |\bdry P_i|\le O(r)\cdot \sum_{i=1}^k |\bdry P_i| = \OO(nr^{\beta})\]
Step 5b can be implemented in the following fashion: when computing the neighborhood balls of various radii for a vertex $v$, go in order of increasing radii and only add vertices to the maintained $\lambda$-interval representation. This can be done incrementally one vertex at a time using the distance information we computed in step 3.
Merging any inserted vertex with any adjacent intervals takes $O(1)$ time per insertion. The total amount of time taken for this step at most $\OO(n^2/r^{1-\beta})$.
Step 5c can be implemented for a vertex $u\in V_{in}\cap P_i$ by sorting the end points of the intervals of $\Rep_\lambda(N_G^{\ell_u - d_G(u\to v)}[v])$ for all $v\in \bdry P_i$ and doing a line sweep to merge the intervals. The $N_{G[P_i]}^{\ell_u}[u]$ term can be directly computed from step 2 and added to the representation since there are at most $r$ vertices.
By the guarantee of \Cref{lem:lowcrossing}, 
this takes time at most 
\begin{align*}
 \sum_{i=1}^{k}
\sum_{u\in P_i\cap V_{in}}
\OO\left(|P_i|+ \sum_{v\in \bdry P_i} |\Rep_\lambda(N_G^{\ell_u-d_G(u\to v)}[v])|\right) 
&= \OO(nr) + \sum_{B\in \cB_i^{out}} \OO(|\Rep_\lambda(B)|)  \\
&= \OO(nr + mn^{1-1/d}) \\
&= \OO(nr +n^{2-1/d}r^{\beta}).
\end{align*}
In step 5d, the sampling can be implemented for a vertex $u\in V_{in}\cap P_i$ in time proportional to the complexity of the union of intervals computed in 5c, which is $\OO(n^{2-1/d}r^\beta)$ time. Computing the distances to its witness using \Cref{eq:dist} for each vertex takes a total of $\OO(nr^{\beta})$ time by the same analysis as for $m$ using the distances computed in step 2 and step 3.

The total running time is minimized by setting $r=n^{1/d}$.
\[ \OO(n^{1+o(1)}+nr + n^2/r^{1-\beta} +n^{1+1/d} + n^{2-1/d}r^{\beta}) = \OO(n^{2-(1-\beta)/d}). \]
We conclude the following theorem.
\thmdiam*



\subsection{Variations}
\label{ssec:variations}
\paragraph*{Bounded distance VC-dimension.}  
If the graph $G$ does not have bounded generalized distance VC-dimension, but instead bounded distance VC-dimension, we can still bound the dual distance VC-dimension (and thus shatter dimension) by \Cref{lem:vcdim-facts}(3) by $2^{d+1}$. This still leads to a truly subquadratic-time algorithm for real-weighted eccentricities, although the running time becomes $\OO(n^{2-(1-\beta)/2^{d+1}})$.

\paragraph*{Farthest reachable vertex.}
If the graph $G$ is not strongly connected, we can still compute the farthest reachable vertex from every vertex in the same running time. To do this, we need to modify step 5(d). For a vertex $u$, instead of checking if the $\lambda$-interval representation $\Rep_\lambda(N_G^{\ell_u}[u])$ consists of all vertices, we need to check if it consists of all reachable vertices from $u$, and sample only among the reachable vertices. The reachable vertices from $u$ can be thought of as $N_G^{\infty}[u]$ -- the set of vertices reachable from $u$ at an arbitrary distance from $u$. A $\lambda$-interval representation $\Rep(N_G^{\infty}[u])$ can be computed from the set of vertices within the piece $P_i$ that $u$ is in (at most $r$ many vertices), as well as $N_G^\infty[v]$ for all vertices $v\in \bdry P_i$. If we add $N_G^\infty[v]$ to $\cD$ before the computation of $\lambda$, they would have low average crossing number, and would allow us to compute the $\lambda$-interval representation of the set of reachable vertices $\Rep_\lambda(N^\infty_G[u])$ using the following identity: 
\[N^{\infty}_G[u] = N_{G[P]}^{\infty}[u] \cup \bigcup_{v\in \bdry P} N^\infty_G[v]\]
Now we can efficiently sample from $N_G^{\infty}[u] \setminus N_G^{\ell_u}[u]$ by a line sweep through the end points of the $\lambda$-interval representations of the two. The overall runtime is unchanged.

\section{Conclusion}
In this paper, we have shown that VC-dimension techniques can be used to compute real-weighted diameter in truly subquadratic time in minor-free graphs and other graphs that have sublinear-sized balanced separators and bounded distance VC-dimension.
One downside is that the algorithm does not work for computing other quantities such as Wiener index. The algorithm described in~\cite{DucoffeHV22} faces the same issue.

Another downside of the result is the reliance on sublinear-sized balanced separators.
The results of~\cite{ChanCGKLZ25} bypass the need for sublinear-sized balanced separators in the unweighted setting by using low diameter decompositions. 
This relied on an alternate algorithm for diameter that runs in truly subquadratic time in unweighted graphs of small diameter.
However, it is unclear if such a low diameter algorithm exists in the real-weighted setting.

\bibliography{ref}
\bibliographystyle{alphaurl}

\end{document}